\documentclass[aps,nopacs,nokeys,superscriptaddress,11pt,twoside]{revtex4}

\usepackage{graphicx,epic,eepic,epsfig,amsmath,latexsym,amssymb,verbatim,revsymb}

\usepackage{theorem}
\newtheorem{definition}{Definition}

\newtheorem{theorem}[definition]{Theorem}
\newtheorem{corollary}[definition]{Corollary}

\def\squareforqed{\hbox{\rlap{$\sqcap$}$\sqcup$}}
\def\qed{\ifmmode\squareforqed\else{\unskip\nobreak\hfil
\penalty50\hskip1em\null\nobreak\hfil\squareforqed
\parfillskip=0pt\finalhyphendemerits=0\endgraf}\fi}
\def\endenv{\ifmmode\;\else{\unskip\nobreak\hfil
\penalty50\hskip1em\null\nobreak\hfil\;
\parfillskip=0pt\finalhyphendemerits=0\endgraf}\fi}
\newenvironment{proof}{\noindent \textbf{{Proof~} }}{\qed}
\newenvironment{remark}{\noindent \textbf{{Remark~}}}{}

\mathchardef\ordinarycolon\mathcode`\:
\mathcode`\:=\string"8000
\def\vcentcolon{\mathrel{\mathop\ordinarycolon}}
\begingroup \catcode`\:=\active
  \lowercase{\endgroup
  \let :\vcentcolon
  }

\newcommand{\nc}{\newcommand}
\nc{\rnc}{\renewcommand}
\nc{\beq}{\begin{equation}}
\nc{\eeq}{{\end{equation}}}
\nc{\beqa}{\begin{eqnarray}}
\nc{\eeqa}{\end{eqnarray}}
\nc{\lbar}[1]{\overline{#1}}
\nc{\bra}[1]{\langle#1|}
\nc{\ket}[1]{|#1\rangle}
\nc{\ketbra}[2]{|#1\rangle\!\langle#2|}
\nc{\braket}[2]{\langle#1|#2\rangle}
\nc{\proj}[1]{| #1\rangle\!\langle #1 |}
\nc{\avg}[1]{\langle#1\rangle}
\rnc{\max}{\operatorname{max}}
\nc{\Rank}{\operatorname{Rank}}
\nc{\smfrac}[2]{\mbox{$\frac{#1}{#2}$}}
\nc{\tr}{\operatorname{Tr}}
\nc{\ox}{\otimes}
\nc{\dg}{\dagger}
\nc{\dn}{\downarrow}
\nc{\cA}{{\cal A}}
\nc{\cB}{{\cal B}}
\nc{\cC}{{\cal C}}
\nc{\cD}{{\cal D}}
\nc{\cE}{{\cal E}}
\nc{\cF}{{\cal F}}
\nc{\cG}{{\cal G}}
\nc{\cH}{{\cal H}}
\nc{\cI}{{\cal I}}
\nc{\cJ}{{\cal J}}
\nc{\cK}{{\cal K}}
\nc{\cL}{{\cal L}}
\nc{\cM}{{\cal M}}
\nc{\cN}{{\cal N}}
\nc{\cO}{{\cal O}}
\nc{\cP}{{\cal P}}
\nc{\cR}{{\cal R}}
\nc{\cS}{{\cal S}}
\nc{\cT}{{\cal T}}
\nc{\cX}{{\cal X}}
\nc{\cZ}{{\cal Z}}
\nc{\csupp}{{\operatorname{csupp}}}
\nc{\qsupp}{{\operatorname{qsupp}}}
\nc{\var}{\operatorname{var}}
\nc{\rar}{\rightarrow}
\nc{\lrar}{\longrightarrow}
\nc{\polylog}{\operatorname{polylog}}
\nc{\1}{{\openone}}

\def\a{\alpha}
\def\b{\beta}

\nc{\RR}{{{\mathbb R}}}
\nc{\CC}{{{\mathbb C}}}
\nc{\FF}{{{\mathbb F}}}
\nc{\NN}{{{\mathbb N}}}
\nc{\ZZ}{{{\mathbb Z}}}
\nc{\PP}{{{\mathbb P}}}
\nc{\QQ}{{{\mathbb Q}}}
\nc{\UU}{{{\mathbb U}}}
\nc{\EE}{{{\mathbb E}}}
\nc{\id}{{\operatorname{id}}}

\nc{\be}{\begin{equation}}
\nc{\ee}{{\end{equation}}}
\nc{\bea}{\begin{eqnarray}}
\nc{\eea}{\end{eqnarray}}
\nc{\<}{\langle}
\rnc{\>}{\rangle}
\nc{\Hom}[2]{\mbox{Hom}(\CC^{#1},\CC^{#2})}
\nc{\rU}{\mbox{U}}

\nc{\ob}[1]{#1}

\newcommand{\norm}[1]{\lVert#1\rVert}

\newcommand{\cocon}[1]{\overline{#1}}

\newcommand{\hull}{{\operatorname{semi-lin}\,}}
\newcommand{\tLambda}{{\widetilde\Lambda}}
\newcommand{\cptn}{{c.p.t.$\leq${}}}

\begin{document}

\title{Non-malleable encryption of quantum information}

\author{Andris Ambainis}
\affiliation{Department of Computer Science, University of Latvia, Raina bulv. 19, Riga,
LV-1586, Latvia}
\affiliation{Department of Combinatorics and Optimization \&{} Institute for Quantum Computing, University of Waterloo}

\author{Jan Bouda}
\affiliation{Faculty of Informatics, Masaryk University, Botanick\'{a} 68a, 602\,00 Brno, Czech Republic}

\author{Andreas Winter}
\affiliation{Department of Mathematics, University of Bristol, Bristol BS8 1TW, U.K.}
\affiliation{Centre for Quantum Technologies, National University of Singapore,
 3 Science Drive 2, Singapore 117543}

\date{3 February 2009}

\begin{abstract}
We introduce the notion of \emph{non-malleability} of a
quantum state encryption scheme (in dimension $d$):
in addition to the requirement that an adversary
cannot learn information about the state, here we demand that no
controlled modification of the encrypted state can be effected.

We show that such a scheme is equivalent to a \emph{unitary 2-design}
[Dankert \emph{et al.}], as opposed to normal encryption which is a
unitary 1-design. Our other main results include a new proof of the lower
bound of $(d^2-1)^2+1$ on the number of unitaries in a 2-design
[Gross \emph{et al.}], which lends itself to a generalization to
approximate 2-design.
Furthermore, while in prime power dimension there is a unitary 2-design
with $\leq d^5$ elements, we show that there are always approximate
2-designs with $O(\epsilon^{-2} d^4 \log d)$ elements.
\end{abstract}

\maketitle

\section*{INTRODUCTION}

The ordinary (and in terms of secret key length, optimal) encryption of quantum
states on $n$ qubits is by applying a randomly chosen tensor product
of Pauli operators (including the identity).
This requires $2n$ bits of shared secret randomness, corresponding to the
$4^n$ Pauli operators.
(More generally, for states on a $d$-dimensional system, one can use
the elements of the discrete Weyl group -- up to global phases -- of
which there are $d^2$.)
This is perfectly secure in the sense that the state the adversary
can intercept is, without her knowing the key, always the maximally mixed state.
For perfectly secure encryption with random unitaries, it was shown
in~\cite{Ambainis+Mosca...-Priva_quant_chann:2000}
that $2n$ bits of secret key are also necessary
for $n$ qubits.
The lower bound of $2$ bits of key per qubit continues to hold even
for $\epsilon$-approximate encryption (up to expressions in $\epsilon$),
but there it becomes relevant how the approximation is defined
--- whether it randomizes entangled states or not [see
Eq.~\eqref{eq:enc-scheme-2-correct} and \eqref{eq:enc-scheme-2-naive} below].
In~\cite{Hayden+Leung...-Rando_quant_state:2003} it was shown that in the
latter case one gets away with $n+o(n)$ key bits for arbitrary $n$-qubit
states; their construction was derandomized later
in~\cite{Ambainis.Smith-Smallpseudo-randomfamilies-2004}
and~\cite{Dickinson.Nayak-ApproximateRandomizationof-2006}.

However, even perfectly secure encryption allows for a different
sort of intervention by the adversary: she can, without ever attempting
to learn the message, change the plaintext
by effecting certain dynamics on the encrypted state. Consider briefly
the classical one-time pad, i.e. an $n$-bit message XORed with a random
$n$-bit string: by flipping a bit of the ciphertext, an adversary
can effectively flip any bit of the recovered plaintext.
In the quantum case, due to the (anti-)commutation relations of the
Pauli operators, by applying to the ciphertext (encrypted state)
some Pauli, she forces that the decrypted state is the plaintext modified
by that Pauli: for an $n$-qubit state $\ket{\varphi}$, any adversary's
Pauli operator $Q$ and secret key Pauli $P_k$, the decrypted state is
\[
  P_k^\dagger Q P_k \ket{\varphi} = \zeta Q\ket{\varphi},
\]
with some (unimportant) global phase $\zeta = \zeta(P,Q)$.

This is evidently an undesirable property of a encryption scheme,
and can be classically addressed e.g. by authenticating the message as well as
encrypting it. Interestingly, in the above quantum message case,
it was shown in~\cite{Barnum+Crepeau...-Authenticatio_of_q_mes:2002}
that authenticating quantum messages is
at least as expensive as encrypting them (it actually encrypts
the message as well): one needs $2$ bits of shared
secret key for each qubit authenticated, even in the approximate setting
considered in~\cite{Barnum+Crepeau...-Authenticatio_of_q_mes:2002}. 
Classical non-malleable cryptosystems include both symmetric and 
asymmetric encryption schemes, bit commitment, zero knowledge proofs 
and others~\cite{DDN01}.

Here we will introduce a formal definition of perfect non-malleability of
a quantum state encryption scheme (NMES), i.e. resistance against predictable
modification of the plaintext,
as well as of two notions of approximate encryption with approximate
non-malleability. We show that a unitary non-malleable channel
is equivalent to unitary $2$-design in the sense of Dankert
\emph{et al.}~\cite{Dankert.Cleve.ea-ExactandApproximate-2006}.
We use this fact to design an
exact ideal non-malleable encryption scheme requiring $5\log d$ bits of key.
Also, the lower bound of
Gross \emph{et al.}~\cite{Gross.Audenaert.ea:Evenlydistributedunitaries:-2007}
for unitary $2$-designs applies for perfect NMES;
we give a new proof of their result that at least
$(d^2-1)^2+1$ unitaries are required, which also yields a
more general lower bound of $(4-O(\epsilon)) \log d$
on the \emph{entropy} of an approximate unitary $2$-design.
Finally we demonstrate that approximate NMES (unitary 2-designs) exist
which require only $4\log d+\log\log d+ O(\log 1/\epsilon)$ bits of key.

\section{General Model of Encryption}

Suppose Alice wants to send a secret quantum message to Bob, say an arbitrary state
$\rho \in \cB(\cH)$, a Hilbert space of dimension $d$. For this purpose they
will use a encryption scheme with pre-shared secret key $K$ as follows.
$K$ is distributed according to some probability distribution $p_K(k)$
and for each $k$ there is a pair of c.p.t.p.~(completely positive and trace
preserving) maps
\[
  E_k:{\cal B}({\cal H}) \longrightarrow {\cal B}({\cal H}')
   \text{ and }
  D_k:{\cal B}({\cal H}') \longrightarrow {\cal B}({\cal H})
\]
for encryption and decryption. The combined effect of en- and decryption, averaged over all
keys, is described by a c.p.t.p. map (noisy quantum channel)
$R:{\cal B}({\cal H}) \longrightarrow {\cal B}({\cal H})$,
acting on operators on ${\cal H}$ as
\[
  R(\rho) = \sum_k p_K(k) D_k\bigl( E_k(\rho) \bigr).
\]
Similarly, for an adversary who intercepts the encrypted state but doesn't know the
secret key, we have an average channel
$R':{\cal B}({\cal H}) \longrightarrow {\cal B}({\cal H}')$,
\[
  R'(\rho) = \sum_k p_K(k) E_k(\rho).
\]
Loosely speaking, the quality of the scheme is described by two parameters:
first, the reliability, i.e.~how close $R$ is to the ideal channel; secondly,
the secrecy, i.e.~how close $R'$ is to a constant (meaning a map taking all
input states to a fixed output state).
In an ideal scheme, $R=\id$ and $R'=\text{const.}$, i.e.~there is a
state $\xi_0$ on ${\cal H}'$, such that
\begin{align}
  \label{eq:ideal-enc-scheme-1}
  \forall \rho &\quad R(\rho) = \rho, \\
  \label{eq:ideal-enc-scheme-2}
  \forall \rho &\quad R'(\rho) = \xi_0.
\end{align}

The issue of approximate performance is a little bit tricky: whereas
for the reliability of communication there is essentially one
notion, namely, for $\delta > 0$,
\begin{equation}
  \label{eq:enc-scheme-1}
  \forall \rho \quad  \norm{\rho - R(\rho)}_1 \le \delta,
  \tag{1'}
\end{equation}
there are two asymptotically radically different notions of secrecy.
One is the ``naive'' one
\begin{equation}
  \label{eq:enc-scheme-2-naive}
  \forall \rho \quad  \bigl\| R'(\rho) - \xi_0 \bigr\|_1 \leq \epsilon
  \tag{2'}
\end{equation}
that does not randomize entangled states when applied locally.

The ``correct'' (composable!) definition takes into account the possibility to
apply $R'$ to part of an entangled state:
\begin{equation}
  \label{eq:enc-scheme-2-correct}
  \forall \rho_{12} \quad
            \bigl\| (R'\ox\id)\rho_{12} - \xi_0\ox\rho_2 \bigr\|_1 \leq \epsilon.
  \tag{2''}
\end{equation}
We note that the two conditions coincide in the ideal case $\epsilon=0$.

The minimal key length required for (approximate) encryption reflects
whether Eq.~\eqref{eq:enc-scheme-2-naive} or Eq.~\eqref{eq:enc-scheme-2-correct}
is used. In the former case $\log d$ bits of key are necessary,
and $\log d+o(\log d)$ bits of key are
sufficient~\cite{Hayden+Leung...-Rando_quant_state:2003,Ambainis.Smith-Smallpseudo-randomfamilies-2004}
to randomize quantum system of dimension $d$, while in the latter case
the key length essentially coincides with the exact encryption case and equals
$(2-O(\epsilon))\log d$~\cite{Ambainis+Mosca...-Priva_quant_chann:2000}.

\section{Non-malleability}

There is, of course, a simple scheme of encryption that implements an
ideal scheme: on $n$ qubits, use a key of length $2n$ and apply an independent
random Pauli operator to each qubit. (More generally, in dimension $d$,
the key identifies one of the $d^2$ discrete Weyl operators
made up of the basis shift and phase shift operators.)
The adversary evidently cannot see any information about the plaintext
state, but she can use the ciphertext in another way: by modulating the ciphertext with an arbitrary
Pauli operation, she can effectively implement this Pauli transformation
on the plaintext state.

We shall show that this is not at all a necessary feature of any encryption
scheme. There are, however, always two possible actions for the adversary (and their
arbitrary convex combination). Namely, not to interfere at all, resulting in
correct decryption of the state $\rho$ sent; or interception
of the ciphertext and its replacement by a state $\eta_0$ on ${\cal H}'$,
resulting in Bob always decrypting the constant state
$\rho_0 = \sum_k p_K(k) D_k(\eta_0)$.
In other words, assuming the adversary implements an arbitrary quantum
channel, i.e.~a completely positive and trace non-increasing (\cptn) map
$\Lambda:{\cal B}({\cal H}') \longrightarrow {\cal B}({\cal H}')$, the class of
\emph{effective channels} on the plaintext she can realize, namely all channels
\begin{equation*}\begin{split}
  \tLambda: {\cal B}({\cal H}) &\longrightarrow {\cal B}({\cal H}) \text{ s.t.}\\
                \rho &\longmapsto     \sum_k p_K(k) D_k\Bigl( \Lambda\bigl( E_k(\rho) \bigr) \Bigr),
\end{split}\end{equation*}
will include all convex combinations of the identity (up to approximation
as specified by $\epsilon$) and the completely forgetful channels $\langle \rho_0 \rangle$
mapping all inputs to the state $\rho_0 = \sum_k p_K(k) D_k(\eta_0)$, with arbitrary $\eta_0$.

We call an encryption scheme \emph{(perfectly) non-malleable}, if these are the only
effective channels the adversary can realize, i.e.~if for every $\Lambda$,
$\tLambda$ is in the semi-linear span of $\id$ and the $\langle \rho_0 \rangle$,
\begin{equation}
  \label{eq:ideal-tres-3}
  \tLambda \in {\cal C} :=
  \hull\left( \{\id\} \cup
                    \left\{ \langle \rho_0 \rangle : \rho \mapsto \rho_0
                                                           = \sum_k p_K(k) D_k(\eta_0) \right\}
             \right),
\end{equation}
with $\hull$ being the semi-linear hull, i.e. with any family of
elements it also contains all their linear combinations,
subject to complete positivity of the resulting operator.
[Clearly, in the above the convex hull can be realized by an adversary; however,
in general the full semi-linear hull is accessible; e.g.~for the Haar measure
on the unitary group -- and infinite key -- the only constant channel is
$\langle \tau \rangle$, with the maximally mixed state $\tau = \frac{1}{d}\1$,
cf.~the beginning of the next section, in
particular eqs.~(\ref{eq:channels})--(\ref{eq:semilinear-example})
On the other hand, any traceless unitary by the adversary results in the 
effective channel $\tLambda(\rho) = \frac{1}{d^2-1}\left(d^2\tau-\rho\right)$.]

Also, a word on why we demand this for all \cptn\ maps, which is a strictly
larger class than c.p.t.p.: note that the adversary could implement an
\emph{instrument}~\cite{DaviesLewis:operational}, which is a resolution
of a c.p.t.p.~map into \cptn\ ones. One of them will act randomly, but
the adversary can learn which one, so could effectively correlate herself
with the effective channel $\tLambda$.

As before, this is to be understood up to approximations: for every effective
channel $\tLambda$ there is $\Theta \in {\cal C}$ such that
\begin{equation}
  \label{eq:tres-3-naive}
  \forall \rho \quad
            \bigl\| \tLambda(\rho) - \Theta(\rho) \bigr\|_1 \leq \theta.
  \tag{3'}
\end{equation}
However, again the ``correct'' (composable) definition has to take into account the possibility
of applying the effective channels to part of an entangled state:
\begin{equation}
  \label{eq:tres-3-correct}
  \forall \rho_{12} \quad
            \bigl\| (\tLambda\ox\id)\rho_{12}
                         - (\Theta\ox\id)\rho_{12} \bigr\|_1 \leq \theta.
  \tag{3''}
\end{equation}
We call the scheme \emph{strictly non-malleable}, if Eq.~(\ref{eq:ideal-tres-3})
or (\ref{eq:tres-3-naive}) or (\ref{eq:tres-3-correct}) holds for some set
${\cal C}' = \hull\bigl\{ \id, \langle\rho_0\rangle \bigr\}$ instead of ${\cal C}$.
(In other words, there is essentially only one constant channel in ${\cal C}$,
independent of $\eta_0$.)
Perfect non-malleability then corresponds to $\theta = 0$, in either
Eq.~(\ref{eq:tres-3-naive}) or (\ref{eq:tres-3-correct})

\section{Main Results}

In this paper we restrict ourselves to the ``minimal''
case, when ${\cal H}' = {\cal H}$ is a $d$-dimensional Hilbert
space, and to perfect transmission,
i.e.~Eq.~(\ref{eq:ideal-enc-scheme-1}).
This entails that $E_k$ is
conjugation by a unitary $U_k$, while $D_k$ is simply the inverse,
i.e.~conjugation by $U_k^\dagger$:
\[
  E_k(\rho) = U_k \rho U_k^\dagger,\quad
  D_k(\sigma) = U_k^\dagger \sigma U_k.
\]

Since convex combinations of unitary conjugation channels are unital,
in an encryption scheme all input states are encrypted as
the maximally mixed state $\xi_0 = \tau := \frac{1}{d}\1$ in
Eqs.~(\ref{eq:ideal-enc-scheme-2}), (\ref{eq:enc-scheme-2-naive})
and (\ref{eq:enc-scheme-2-correct}).
(For a more general discussion
see~\cite{BoudaZiman:Optimalityofprivate-2007}.)
This means that the adversary can always implement channels
\begin{equation}
  \label{eq:channels}
  \Theta \in \cC' = \hull\{ \id, \langle \tau \rangle \},
\end{equation}
where $\langle \tau \rangle$ is the completely depolarizing channel.
Conversely, we demand that these are the only ones she can
achieve: for every \cptn\ map $\Lambda$, we demand that the effective
channel $\tLambda \in \cC'$, with
\[
  \tLambda(\rho) = \sum_k p_K(k) U_k^\dagger \bigl( \Lambda(U_k \rho U_k^\dagger) \bigr) U_k.
\]

This can be conveniently re-expressed using the Choi-Jamio\l{}kowski
operators~\cite{Choi:matrix,Jamiolkowski-Lineartransformationswhich-1972}:
for the maximally entangled state
$\Phi_d = \frac{1}{d}\sum_{i,j=0}^{d-1}\ket{ii}\!\bra{jj}$ on two systems
labelled $1$ and $2$, let $\omega = J_\Lambda := (\Lambda \otimes \id)\Phi_d$.
Note that $\tr J_\Lambda \leq 1$ and that $\Lambda$ can be be recovered
from the Choi-Jamio\l{}kowski operator as follows:
\begin{equation}
  \label{eq:CJ-inverse}
  \Lambda(\rho) = d \tr_2\bigl( (\1\otimes\rho^\top)J_\Lambda \bigr),
\end{equation}
where $\rho^\top$ is the transpose operator of $\rho$ with respect to
the basis $\{ \ket{i} \}_{i=0}^{d-1}$.
The image of the set $\cC'$ under the Choi-Jamio\l{}kowski isomorphism
is the set of bipartite positive operators
\begin{equation}
  \label{eq:semilinear-example}
  (\cC' \otimes \id)\Phi_d = \hull\{ \Phi_d, \tau\otimes\tau \}
                           = \RR_{\geq 0}\Phi_d + \RR_{\geq 0}(\1-\Phi_d) =: \cI,
\end{equation}
which are (up to normalization) just the so-called \emph{isotropic states}.
Note that these are exactly the (semidefinite) operators invariant
under conjugation with $U\ox\cocon{U}$, and that integration over the
Haar measure ${\rm d}U$ implements the projection into $\cI$: for
every operator $X$,
\begin{equation}
  \label{eq:iso-twirl}
  \int {\rm d}U (U\ox\cocon{U}) X (U\ox\cocon{U})^\dagger = \a\Phi_d + \b(\1-\Phi_d),
  \text{ with }
  \a = \tr X\Phi_d,\ \b=\frac{1}{d^2-1}\tr X(\1-\Phi_d).
\end{equation}
The c.p.t.p.~mapping from $X$ to the above average is known as the $U\ox\cocon{U}$-twirl,
denoted $\cT_{U\ox\cocon{U}}$.

On the other hand, exploiting the symmetry
$\Phi_d = (U\ox\cocon{U})\Phi_d(U\ox\cocon{U})^\dagger$,
we can write the Choi-Jamio\l{}kowski operator of the effective channel,
\[\begin{split}
  \widetilde\omega &= (\tLambda \otimes \id)\Phi_d \\
                   &= \sum_k p_K(k) (U_k\ox\1)^\dagger
                                      \Bigl[ (\Lambda\ox\id)
                                              \bigl( (U\ox\1)\Phi_d(U\ox\1)^\dagger\bigr) \Bigr]
                                   (U_k\ox\1)      \\
                   &= \sum_k p_K(k) (U_k\ox\1)^\dagger
                                      \Bigl[ (\Lambda\ox\id)
                                        \bigl( (\1\ox U_k^\top)\Phi_d(\1\ox U_k^\top)^\dagger\bigr) \Bigr]
                                   (U_k\ox\1)      \\
                   &= \sum_k p_K(k) (U_k\ox\cocon{U_k})^\dagger
                                      \bigl[ (\Lambda\ox\id) \Phi_d \bigr]
                                    (U_k\ox\cocon{U_k}) \\
                   &= \sum_k p_K(k) (U_k\ox\cocon{U_k})^\dagger \omega (U_k\ox\cocon{U_k})
                    =: \cT(\omega),
\end{split}\]
where $\cT$ is manifestly a c.p.t.p.~map.
The condition that $\{ p_K(k), U_k \}$ forms a perfect
NMES is now concisely expressed as $\cT = \cT_{U\ox\cocon{U}}$.

This is precisely the condition for a so-called \emph{unitary 2-design}
\cite{Dankert.Cleve.ea-ExactandApproximate-2006}, see also
\cite{Gross.Audenaert.ea:Evenlydistributedunitaries:-2007}.
Note that modulo a partial transpose, the $U\ox\cocon{U}$-twirl
is equivalent to the more familiar $U\ox U$-twirl
\[
  \cT_{U\ox U}(X) = \int {\rm d}U (U\ox U) X (U\ox U)^\dagger
                  = \a F + \b (\1-F),
\]
with the swap (or flip) operator $F = \sum_{i,j=0}^{d-1} \ket{ij}\!\bra{ji}$,
mapping density operators to \emph{Werner states}~\cite{Werner-QuantumStateswith-1989}.
Thus we have proved,
\begin{theorem}
  \label{thm:TRES-is-2design}
  Every perfect non-malleable encryption scheme is a unitary $2$-design.
  \qed
\end{theorem}

\begin{corollary}
  \label{cor:TRES-implies-encryption}
  Any perfect non-malleable encryption scheme, i.e., an ensemble of
  unitaries $\{ p_K(k), U_k \}$ satisfying $\tLambda \in \cC'$, is automatically
  an ideal encryption scheme, i.e.~Eq.~(\ref{eq:ideal-enc-scheme-2})
  holds.
\end{corollary}
\begin{proof}
  By Theorem~\ref{thm:TRES-is-2design} a perfect NMES is a unitary
  $2$-design. But then it is automatically a unitary $1$-design,
  meaning that for all $\rho$, $\sum_k p_K(k) U_k \rho U_k^\dagger = \tau$,
  which is precisely Eq.~(\ref{eq:ideal-enc-scheme-2}).
\end{proof}

\begin{theorem}
  \label{thm:lowerbounds}
  Every perfect non-malleable encryption scheme $\{ p_K(k), U_k \}$
  requires at least $(d^2-1)^2+1$ unitaries. Furthermore, every
  $\theta$--NMES as in Eq.~(\ref{eq:tres-3-correct}) with $\theta \leq 1/e$
  satisfies
  \[
    H(p_K) \geq H_2\left(\frac{1}{d^2}\right) + 2\left(1-\frac{1}{d^2}\right)\log(d^2-1)
                - 4\theta\log d - H_2(\theta)
           \geq (4-O(\theta)) \log d,
  \]
  where $H_2(x) = -x\log x - (1-x)\log(1-x)$ is the binary entropy.
\end{theorem}

\begin{remark}
  In the light of Theorem~\ref{thm:TRES-is-2design}, the first part amounts to
  a demonstration that $2$-designs have to have at least $(d^2-1)^2+1$ unitaries; 
  this was proved by Gross
  \emph{et al.}~\cite{Gross.Audenaert.ea:Evenlydistributedunitaries:-2007},
  but we give a different, direct, proof below.
  It seems that it is conjectured that in fact the better lower
  bound $d^2(d^2-1)$ holds in general -- which is true for
  so-called ``Clifford twirls'', and tight in some
  dimensions~\cite{Chau:UnconditionallySecureKey-2005,Gross.Audenaert.ea:Evenlydistributedunitaries:-2007}.
\end{remark}

\medskip
\begin{proof}
Consider the Choi-Jamio\l{}kowski operator of $\cT$,
labeling the systems $1$, $2$, $1'$ and $2'$, and with the maximally
entangled state understood between systems $12$ and $1'2'$:
\[
  \Omega_{U\ox\cocon{U}} :=
  (\cT_{U\ox\cocon{U}}^{12} \ox \id^{1'2'})\Phi_{d^2}
                 = \frac{1}{d^2}\Phi_d^{12} \ox \Phi_d^{1'2'}
                   + \frac{1}{d^2(d^2-1)} (\1-\Phi_d)^{12} \ox (\1-\Phi_d)^{1'2'}.
\]
On the other hand, for the first part of the theorem this has to
be equal to
\[
  \Omega :=
  (\cT^{12}\ox\id^{1'2'})\Phi_{d^2}
        = \sum_{k=1}^N p_K(k) (U_k^1 \ox \cocon{U}_k^2 \ox \1^{1'2'})
                                \Phi_{d^2}
                              (U_k^1 \ox \cocon{U}_k^2 \ox \1^{1'2'})^\dagger.
\]
Comparing ranks of the two right hand side expressions reveals immediately
$N \geq (d^2-1)^2+1$.

\medskip
For the entropy statement in the approximate case, we note that by
Eq.~(\ref{eq:tres-3-correct}),
$\| \Omega - \Omega_{U\ox\cocon{U}} \|_1 \leq \theta$, so by
Fannes' inequality~\cite{Fannes:inequality} and Schur concavity of
the entropy~\cite{Wehrl:review},
\[
  H(p_K) \geq S(\Omega) \geq S(\Omega_{U\ox\cocon{U}}) - \theta\log d^4 - H_2(\theta),
\]
and we are done.
\end{proof}

\begin{theorem}[Chau~\cite{Chau:UnconditionallySecureKey-2005}, 
                Gross \emph{et al.}~\cite{Gross.Audenaert.ea:Evenlydistributedunitaries:-2007}]
  \label{thm:d5}
  If $d=p^n$ is a prime power, then 
  there exists a perfect non-malleable encryption scheme with
  $d^5-d^3$ unitaries, meaning that the key length is $\leq 5 \log d$.
  In fact, such a scheme is obtained as the uniform ensemble
  over a particular subgroup of the Clifford group (i.e., the
  normalizer) of the $n$-th power Heisenberg-Weyl (aka generalised Pauli) 
  group ${\cal P}_p^{\ox n}$, where
  ${\cal P}_p$ is the group generated by the discrete Weyl operators
  \[
    X_p = \sum_{j=0}^{p-1} \ket{j\!+\!1 \mod p}\!\bra{j},\quad
    Z_p = \sum_{k=0}^{p-1} e^{2\pi i k/p}\proj{k}.
  \]
\end{theorem}
\begin{proof}
  Apart from Chau~\cite{Chau:UnconditionallySecureKey-2005} see
  Gross \emph{et al.}~\cite{Gross.Audenaert.ea:Evenlydistributedunitaries:-2007},
  as well as the crisp presentation of Grassl~\cite{Grassl:6-SIC}.
\end{proof}

\medskip
\begin{remark}
  We note that in even prime power dimension, the cardinality
  of the subgroup can be reduced to $(d^5-d^3)/8$.
  Furthermore, Chau~\cite{Chau:UnconditionallySecureKey-2005} showed
  that for several small dimensions the minimum $d^4-d^2$ is attainable;
  see also Gross \emph{et al.}~\cite{Gross.Audenaert.ea:Evenlydistributedunitaries:-2007}
  for another example of $2(d^4-d^2)$.
\end{remark}

\begin{theorem}
  \label{thm:approx-2-design}
  For $0< \theta \leq 1/2$ there exists a $\theta$-NMES
  with $O(\theta^{-2}d^4\log d)$ unitaries, i.e.~with key
  requirement of $4 \log d + \log\log d + O\left(\log\frac{1}{\theta}\right)$
  bits.
  In fact, Eq.~(\ref{eq:tres-3-correct}) holds in the stronger form
  \begin{equation}
    \label{eq:tres-3-strong}
    (1-\theta) \Theta \leq \tLambda \leq (1+\theta) \Theta. \tag{$3^*$}
  \end{equation}
\end{theorem}
\begin{proof}
Start from any exact unitary $2$-design, such as the unitary group with
Haar measure, or the Clifford group or one of its admissible subgroups.
We shall select $U_1,\ldots U_N$ independently at random from that
chosen 2-design, and show that Eq.~(\ref{eq:tres-3-strong}) is true
with high probability as soon as $N \gg \theta^{-2}d^4\log d$; which
of course implies that there exist a particular selection of an ensemble
$\{ 1/N, U_k \}_{k=1}^N$ satisfying (\ref{eq:tres-3-strong}).

In fact, it is sufficient to show that for
$\cT(\omega) = \frac{1}{N} \sum_{k=1}^N (U_k\ox\cocon{U}_k) \omega (U_k\ox\cocon{U}_k)^\dagger$,
\[
  (1-\theta)\cT_{U\ox\cocon{U}} \leq \cT \leq (1+\theta)\cT_{U\ox\cocon{U}},
\]
which in turn is equivalent to the corresponding statement for the
Choi-Jamio\l{}kowski states -- compare Eq.~(\ref{eq:CJ-inverse}):
\[
  (1-\theta)\Omega_{U\ox\cocon{U}} \leq \Omega \leq (1+\theta)\Omega_{U\ox\cocon{U}},
\]
where
\begin{align*}
  \Omega_{U\ox\cocon{U}} &= (\cT_{U\ox\cocon{U}}^{12} \ox \id^{1'2'})\Phi_{d^2}
                          = \frac{1}{d^2}\Phi_d^{12} \ox \Phi_d^{1'2'}
                            + \frac{1}{d^2(d^2-1)} (\1-\Phi_d)^{12} \ox (\1-\Phi_d)^{1'2'}, \\
  \Omega                 &= (\cT^{12}\ox\id^{1'2'})\Phi_{d^2}
                          = \frac{1}{N} \sum_{k=1}^N (U_k^1 \ox \cocon{U}_k^2 \ox \1^{1'2'})
                                                       \Phi_{d^2}
                                                     (U_k^1 \ox \cocon{U}_k^2 \ox \1^{1'2'})^\dagger.
\end{align*}

Now $\Omega$ is a random variable, in fact an average of $N$ independent,
identically distributed terms
\(
 X_k := (U_k^1 \ox \cocon{U}_k^2 \ox \1^{1'2'}) \Phi_{d^2} (U_k^1 \ox \cocon{U}_k^2 \ox \1^{1'2'})^\dagger
\)
with expectation $\EE X_k = \EE \Omega = \Omega_{U\ox\cocon{U}}$. All $X_k$ are
bounded between $0$ and $\1$, so the technical result from~\cite{AhlswedeWinter-ID}
applies, the \emph{operator Chernoff bound}, yielding (with a universal
constant $c>0$)
\[
  \Pr\bigl\{ (1-\theta)\Omega_{U\ox\cocon{U}} \leq \Omega \leq (1+\theta)\Omega_{U\ox\cocon{U}} \bigr\}
       \geq 1 - 2d^4 e^{-c \theta^2 N/d^4},
\]
which implies the claim.
\end{proof}

\section{Discussion}
We have introduced the cryptographic primitive of a non-malleable
quantum state encryption scheme. While many questions remain open,
we have shown that every such scheme based on random unitaries
is a unitary 2-design, showing in particular that every such
scheme must use $4\log d$ bits of key, as opposed to the well-known
$2\log d$ necessary and sufficient for quantum state
encryption~\cite{Ambainis+Mosca...-Priva_quant_chann:2000}.

This situation essentially persists even if we relax the non-malleability
to being approximate. On the other hand, there exists an exact
construction based on the Jacobi subgroup of the Clifford group
in dimension $d$, which requires $5\log d$ bits of key, and
we show a new randomized construction requiring only
$(4+o(1))\log d$ bits of key. We leave open the question of
finding an explicit description of such a scheme, as well
as that of finding an exact unitary 2-design with only $O(d^4)$
elements.

What we also leave open is the perhaps more pressing problem of
relaxing the condition that encryption is done by unitaries. Giving
up this restriction results in an advantage in key size, see the
work of Barnum \emph{et al.}~\cite{Barnum+Crepeau...-Authenticatio_of_q_mes:2002}.
More precisely, these authors show how using $2n+O(s)$ bits of secret
key to encrypt $n-s$ qubits into $n$ qubits results in a $\theta$-NMES
with $\theta = 2^{-O(s)}$. In our setting this can be understood
as only using $d_0 < d$ of the Hilbert space dimensions for quantum
information. Then, to transmit a state in the $d_0$-dimensional space
${\cal H}_0 \subset {\cal H}$, first $s$ key bits are used to
specify a unitary rotation $V_\ell$ of ${\cal H}$, and then the familiar
further $2\log d$ bits of key are used to encrypt ${\cal H}$.
If the $V_\ell$ ($\ell=1,\ldots,2^s$) are ``sufficiently random''
and $2^s \geq d/d_0$ then it can be shown that while the
adversary can implement certain effective channels on ${\cal H}$,
for most $\ell$ this will map the state significantly
outside of ${\cal H}_\ell := V_\ell {\cal H}_0$.

\acknowledgments
JB and AW thank the Perimeter Institute for Theoretical Physics for
its hospitality during a visit in 2006, where the present work was conceived.

%
JB acknowledges support of the Hertha Firnberg ARC stipend program, and
grant projects GA\v{C}R 201/06/P338, GA\v{C}R 201/07/0603 and MSM0021622419.
AW received support from the European Commission (project ``QAP''), from
the U.K.~EPSRC through the ``QIP IRC'' and an Advanced Research Fellowship, and
through a Wolfson Research Merit Award of the Royal Society.
The Centre for Quantum Technologies is funded by the Singapore Ministry of Education
and the National Research Foundation as part of the Research Centres of Excellence programme.

\end{document}